\documentclass[letterpaper, 10 pt, conference]{ieeeconf}  %

\IEEEoverridecommandlockouts                              %
\usepackage{graphics} %
\usepackage{epsfig} %
\usepackage{times} %
\usepackage{amsmath} %
\usepackage{amssymb}  %

\usepackage{amsthm}
\usepackage{bbm}
\usepackage{cite}
\usepackage[overload]{empheq}  %
\usepackage{tabularx}
\usepackage{longtable} %
\usepackage{colortbl}
\usepackage{bbding}
\usepackage{caption} %
\usepackage{comment} %
\usepackage{fancyhdr} %
\usepackage{lipsum} %
\usepackage{tcolorbox} %
\usepackage{stfloats} %
\theoremstyle{definition}
\newtheorem{definition}{Definition}%
\newtheorem{assumption}{Assumption}%
\newtheorem{theorem}{Theorem}%
\newtheorem{lemma}{Lemma}
\newtheorem{remark}{Remark}
\usepackage{nicefrac}
\usepackage{subfig}
\usepackage{transparent}
\usepackage{graphicx}
\usepackage{booktabs}
\usepackage{hyperref}
\usepackage{cleveref}
\usepackage{xcolor}
\usepackage{siunitx}

\newcommand\scalemath[2]{\scalebox{#1}{\mbox{\ensuremath{\displaystyle #2}}}}

\newcommand\real[1]{\mathbb{R}^{#1}}
\newcommand{\preprintswitch}[2]{#2} % arXiv

\newcounter{properties_cayley}

\Crefname{assumption}{Assumption}{Assumptions}

\title{\LARGE \bf
Unconstrained Parametrization of Dissipative and Contracting Neural Ordinary Differential Equations
}

\author{Daniele Martinelli, Clara Luc\'{i}a Galimberti, Ian R. Manchester, Luca Furieri, and Giancarlo Ferrari-Trecate
\thanks{D. Martinelli, C. L. Galimberti, L. Furieri, and G. Ferrari-Trecate are with the Institute of Mechanical Engineering, EPFL, Switzerland. E-mail addresses: \{daniele.martinelli, clara.galimberti, luca.furieri, giancarlo.ferraritrecate\}@epfl.ch.}
\thanks{Ian Manchester is with  the Australian Centre for Robotics (ACFR) and the School of Aerospace, Mechanical and Mechatronic Engineering, University of Sydney,
Australia. E-mail address: ian.manchester@sydney.edu.au.}
\thanks{Research supported by the Swiss National Science Foundation (SNSF) under the NCCR Automation (grant agreement 51NF40\textunderscore 80545). Luca Furieri is also grateful to the SNSF for the Ambizione grant PZ00P2\textunderscore208951.}
}

\begin{document}

\maketitle
\thispagestyle{empty}
\pagestyle{empty}

\begin{abstract}
In this work, we introduce and study a class of Deep Neural Networks (DNNs) in continuous-time. The proposed architecture stems from the combination of Neural Ordinary Differential Equations (Neural ODEs) with the model structure of recently introduced Recurrent Equilibrium Networks (RENs). We show how to endow our proposed NodeRENs with contractivity and dissipativity  --- crucial properties for robust learning and control. Most importantly, as for RENs, we derive parametrizations of contractive and dissipative NodeRENs which are unconstrained, hence enabling their learning for a large number of parameters.
We validate the properties of NodeRENs, including the possibility of handling irregularly sampled data, in a case study in nonlinear system identification.
\end{abstract}
\section{Introduction}
\label{sec:introduction}
Learning complex nonlinear mappings from data is a fundamental challenge in various engineering applications, including computer vision, healthcare, internet %
of things%
, and smart cities~\cite{sarker2021machine}.
Deep Neural Networks (DNNs) have emerged as powerful tools for this task, thanks to their adaptability and ability to generalize their predictions from large amounts of data.
However, DNNs often lack robustness: a slight change in the input data may yield highly different outputs, eventually leading to poor generalization capabilities~\preprintswitch{\cite{szegedy2013intriguing,HaberRuthotto2017Stablearchitecturesfordnn, zakwan2023contractiveHDNNs}}{\cite{szegedy2013intriguing,HaberRuthotto2017Stablearchitecturesfordnn, Cheng_Yi_2020_Seq2Sick, zakwan2023contractiveHDNNs}}. 
This deficiency becomes especially critical when DNNs are implemented in real-world systems like safety-critical power grids or human-interacting robots. In these applications, sensor measurements are inevitably affected by multiple sources of noise and uncertainty, and a lack of robustness may result in substantial losses.

To endow DNN models with formal stability and robustness guarantees, \cite{HaberRuthotto2017Stablearchitecturesfordnn} proposes to equip DNN layers with a dynamical system interpretation. The work \cite{kim2018standard} establishes dynamical DNN models which universally approximate all nonlinear dynamical systems defined in discrete-time and proposes a set of convex constraints that enforce the stability of the DNN model during training. 
In \cite{revay2023recurrent}, the authors have developed discrete-time Recurrent Equilibrium Networks~(RENs) that result from the closed-loop interconnection of a discrete-time linear dynamical system with a static nonlinearity. 
A main contribution of \cite{revay2023recurrent} is to provide an unconstrained parametrization {(also known as \textit{direct} parametrization)} of a class of RENs with properties of stability and dissipativity that are \emph{built-in}, i.e.,  that holds for any choice of the parameters, and without the need to constrain them to a subset. This property enables parameter optimization for very deep models through unconstrained-gradient-descent-based algorithms, while ensuring stability and dissipativity at any iteration.
The usefulness of RENs for system identification and optimal control has been demonstrated in \cite{revay2023recurrent} and \cite{Wang2021YoulaREN}, respectively. However, the corresponding stability and dissipative guarantees are only compatible with discrete-time or sampled-data systems.
Studying how the properties of nonlinear discrete-time dynamics port to their continuous-time counterpart is usually a challenging problem \cite{aastrom1984zeros}.

The recently proposed Neural Ordinary Differential Equations (Neural ODEs)~\cite{chenNeuralOrdinaryDifferential2019} offer a bridge between DNNs and continuous-time dynamics. By unfolding infinitely many DNN layers using a parametrized ODE and employing the \textit{adjoint sensitivity method} for training~\cite{chenNeuralOrdinaryDifferential2019}  , Neural ODEs offer several advantages over discrete DNN models. These include memory efficiency, adaptive computations beyond Euler-like integration methods, and the ability to handle data arriving at arbitrary time instants.

Expanding on Neural ODEs, \cite{kidger2020neural} introduced a continuous-time analogous of Recurrent Neural Networks, albeit without stability or dissipativity guarantees. 
The work \cite{galimberti2023hamiltonian} has proposed stable Hamiltonian Neural ODEs which further exhibit non-vanishing gradient flows. Additionally, \cite{zakwan2023contractiveHDNNs} developed a class of Hamiltonian ODEs with contractivity properties. Furthermore,  \cite{massaroli2020stable} derived architectures ensuring convergence to a stable set. 
However, an architecture combining the robustness and expressiveness of~\cite{kim2018standard,revay2023recurrent} with the benefits of Neural ODEs is not available to the best of our knowledge.
\subsection{Contributions}
In this paper, we establish an alternative version of the REN architecture of \cite{revay2023recurrent} that is compatible with Neural ODEs in continuous-time. Consequently, we call our architectures \textit{NodeRENs}.
After showing that NodeRENs induce {well-posed} dynamical systems in continuous-time, we establish that our models enjoy stability and dissipativity properties {by design}, %
i.e., without any constraints on the space of the parameters that describe the DNNs — akin to their discrete-time counterpart~\cite{revay2023recurrent}. 
NodeRENs, also inherit the advantages of Neural ODEs in continuous-time; they are compatible with any integration scheme that can be chosen based on a trade-off between accuracy and computational resources, and they can be evaluated at any chosen time instant, without the requirement to be uniformly sampled in time.

The paper is structured as follows. In \Cref{sec:preliminary_knowledge}, we describe the problem setting, including the definition of contractivity and Integral Quadratic Constraints~(IQCs). In \Cref{sec:NODEREN}, we introduce the {NodeREN} model and detail the steps to obtain NodeRENs that are contracting and dissipative by design. %
Moreover, in \Cref{sec:Simulations_and_results} the properties and performance of NodeRENs are validated on a nonlinear system identification problem. Finally, in \Cref{sec:Conclusions} we summarize the conclusions of this work, outlining potential future directions for research.

\preprintswitch{For the sake of conciseness, the proofs of all theorems in this work are reported in the Appendix of the full version~\cite{martinelli2023unconstrained}.
}{}
\subsection{Notation}
We denote the set of non-negative real numbers as $\real{+}_0$. 
For $T>0$, let $PC([0,T],\real{n})$ be the space of piecewise-continuous functions in the time interval $[0,T]$.
We represent the Euclidean norm of $v \in \real{n}$ with $|v|$.
For a square matrix $X$, we use the notation $X \succ 0$ $(X \succeq 0)$ and $X\prec 0 $ $ (X\preceq 0) $ to denote positive (semi-) definiteness and negative (semi-) definiteness, respectively. 
The minimum and maximum eigenvalues of the square matrix $X$ are denoted as $\lambda_{min}(X)$ and $\lambda_{max}(X)$, respectively. 
We use $(*)$ to represent a symmetric term in a quadratic expression, e.g., \((X+Y)Q(*)^\top = (X+Y)Q(X+Y)^\top  \), for some $Q \in \real{n \times n}$ and $X,Y \in \real{q \times n}$. With $(*)$, we also indicate elements in symmetric matrices that can be obtained by symmetry.
We use $[{M}]_{p \times m}$ to indicate the block matrix with the first $p$ rows and $m$ columns of ${M}\in\real{n\times n}$ with $n\ge p$ and $n \ge m$.

\section{Problem Formulation}
\label{sec:preliminary_knowledge}
Consider a nonlinear system $\Sigma_\theta$ in continuous-time
\begin{equation}\label{eq:diff_system}
        \Sigma_{\theta} \: = \:
\begin{cases}
    \: \dot{x}(t) = f_{\theta}(x(t),u(t))\\
    \: y(t) = g_{\theta}(x(t),u(t))
\end{cases}\,,
\end{equation}
where $x(t)\in \real{n}$, $y(t)\in \real{p}$ and $u(t)\in \real{m}$ denote the state, output and input of the system at any time $t \in \real{+}_0$, respectively, and $x(0) = x_0$.   In~\eqref{eq:diff_system}, $\theta \in \real{n_\theta}$ denotes a vector of parameters affecting the behavior of the system.
{The idea behind Neural ODEs \cite{chenNeuralOrdinaryDifferential2019} is to interpret some specific DNNs (e.g., ResNets~\cite{he2016deepresnetpaper}) as discretized ODEs. Then, the family of Neural ODEs can be obtained when the discretization step converges to zero, effectively passing from discrete-time DNNs 
to continuous-time models as \eqref{eq:diff_system}}.%
\footnote{For instance, standard ResNet %
models whose hidden states evolve as  $h_{k+1} = h_{k} + \Delta t\, f(h_k,\theta_k)$ can be interpreted as the ODE $\dot{h}(t) = f (h(t),\theta(t))$ as the discretization step $\Delta t$ tends to zero.}
Accordingly, in this paper, we consider the following learning problem:
\begin{equation}\label{eq:minimization_problem}
\begin{aligned}
&\min_{\theta} \quad {L}(\theta,{\mathcal{Z}}) \\
&\textrm{subject to} \;\; \eqref{eq:diff_system}\,.
\end{aligned}
\end{equation}
$L$ is a scalar loss function that can depend on the trajectories ${x}(t)$ and ${y}(t)$ of \eqref{eq:diff_system}, for $t \in [0,T]$ with $T>0$, and ${\mathcal{Z}}$ is a given training dataset. 
For instance, in a classification task, $y(T)$ can play the role of the DNN output layer which is compared to the label $y^{j}_{label}$ of the point $x^j(0)$ contained in a dataset ${\mathcal{Z}} = \left\{(x^j(0),y^{j}_{label})\right\}_{j=1}^N$ for $N\in \mathbb{N}$.
Neural ODEs can exploit any chosen numerical method to perform the forward propagation and the gradient computations through \eqref{eq:diff_system}, including those based on adaptive sampling to guarantee a desired level of precision (e.g., Dormand-Prince, Bogacki–Shampine~\cite{Hairer1993solvingodeI}).
Although possible, back-propagating through numerical solver's operations can be highly expensive in terms of memory and introduce numerical errors. As an alternative, to solve \eqref{eq:minimization_problem} one can use the adjoint method — we refer to \cite{chenNeuralOrdinaryDifferential2019} for full details. %

Neural-ODEs in their general form are not guaranteed to yield stable or dissipative dynamical flows.
Such properties are fundamental in optimal control and system identification, as well as in robust learning problems dealing with noisy features \cite{zakwan2023contractiveHDNNs} and adversarial attacks~\cite{kang2021stable}. 
 Specifically, motivated by \cite{revay2023recurrent}, in this paper, we focus on continuous-time systems that are {contracting} and systems that satisfy incremental~IQCs. We proceed with formally defining both. 
 
 Firstly, given a system $\Sigma_\theta$ with initial condition $a \in \real{n}$ and an input function $u^{[1]} \in PC([0,\infty],\real{m})$, let $x_a^{[1]}$ and $y_a^{[1]}$ denote the corresponding state and output trajectories, respectively.
\begin{definition}%
\label{def:Contractivity_continuous_time}
A system $\Sigma_\theta$ in the form \eqref{eq:diff_system} is said to be \textit{contracting} if for any two initial conditions $a,b \in \real{n}$ and given the same input trajectory ${u}^{[1]} \in PC([0,\infty],\real{m})$, the corresponding state trajectories $x_a^{[1]}$ and $x_b^{[1]}$ satisfy:
\begin{equation}\label{eq:definition_contracting_system}
    | x_a^{[1]}(t) - x_b^{[1]}(t) | \le \kappa e^{-ct} | a - b | \,,
\end{equation}
for all $t\in \real{+}_0$ and for some $c>0$, $\kappa > 0$.
\end{definition}
The \Cref{def:Contractivity_continuous_time} can be interpreted as follows: a contracting system `forgets' the initial condition exponentially fast as time progresses. Hence, all trajectories converge to each other, independently of the initial state.

Next, we define dissipative systems that satisfy incremental IQCs. 
Dissipative systems cannot increase their internal energy despite external inputs (e.g., feedback control actions or disturbances). As such, they can be designed to possess finite input-output gains, a crucial property in nonlinear control theory~\cite{van2000l2}, as well as robust learning~\cite{pauli2021training}. In order to formally define these properties, let $a,b \in \real{n}$ be two initial conditions and $u^{[1]},u^{[2]} \in PC([0,\infty],\real{m})$ be two input trajectories. Define the relative displacements as
\begin{equation}\label{eq:delta_definition_y_u_x}
    \begin{gathered}
        \Delta y(t) = y_{a}^{[1]}(t) - y_{b}^{[2]}(t) \: , \: \Delta u(t) = u^{[1]}(t) - u^{[2]}(t)\: , \\
    \Delta x(t) = x_{a}^{[1]}(t) - x_{b}^{[2]}(t) \: .
    \end{gathered}
\end{equation}
Let
\begin{equation}\label{eq:definition_supply_rate_Q_S_R}
    s_\Delta(\Delta u(t),\Delta y(t)) =\begin{bmatrix}
    \Delta y(t)\\
    \Delta u(t)
    \end{bmatrix}^\top 
    \begin{bmatrix}
    Q & S^\top \\
    S & R
    \end{bmatrix}
    \begin{bmatrix}
    \Delta y(t)\\
    \Delta u(t)
    \end{bmatrix}  \,,  
\end{equation}
be a quadratic function 
\(
     s_\Delta  :  \real{m} \times \real{p} \rightarrow \real{}
\) %
parametrized by $Q \in \real{p \times p}$, $S \in \real{m \times p}$, $R \in \real{m \times m}$. We recall the definition of dissipative systems that satisfy IQCs according to the supply rate \eqref{eq:definition_supply_rate_Q_S_R}.
\begin{definition}%
\label{def:Incremental_IQC}
A system $\Sigma_\theta$ in form \eqref{eq:diff_system} satisfies the 
\textit{incremental} IQC defined by the matrices $(Q,S,R)$, with $Q \preceq 0$ and $R = R^\top$,
if there exists a function $\mathcal{S}:\real{n} \to \mathbb{R}^+_0$ such that, 
for any {two} initial conditions $a,b\in \real{n}$ and any two possible input functions $u^{[1]},u^{[2]} \in PC([0,\infty],\real{m})$, 
\begin{equation}
    \label{eq:definition_IQC_continuous_time}
    \mathcal{S}(\Delta x(t_1)) \le \mathcal{S}(\Delta x(t_0)) 
    + \int_{t_0}^{t_1} s_\Delta(\Delta u(t),\Delta y(t))dt \,,
\end{equation}
for every $t_1 \geq t_0$, where $\Delta x$, $\Delta y$ and $\Delta u$ are defined in \eqref{eq:delta_definition_y_u_x} and the supply rate $s_\Delta$ is defined in \eqref{eq:definition_supply_rate_Q_S_R}.
\end{definition}

Despite being restricted to quadratic supply functions according to \eqref{eq:definition_supply_rate_Q_S_R}, IQCs can certify many incremental properties by appropriately selecting the values of $(Q,S,R)$. It is worth noting that the assumptions in \Cref{def:Incremental_IQC}
\begin{equation}\label{eq:Assumptions_Q_R}
    Q \preceq 0 \,, \quad R = R^\top\,,
\end{equation}
are fulfilled for several incremental properties of interest, see \Cref{tab:choice_Q_S_R}.
\begin{table}[ht]
\centering
\begin{tabular}{@{}l|ccc|l@{}}
\toprule
Incremental Property & $Q$ & $R$ & $S$             & $s_\Delta(\Delta u, \Delta y)$ \\ \midrule
$L_2$- gain bound $(\gamma\ge0)$& $- \frac{1}{\gamma}I$ & $\gamma I$ & $0$ & $\gamma | \Delta u |^2 - \frac{1}{\gamma} | \Delta y |^2$ \\
Passivity            & $0$ & $0$ & $\frac{1}{2} I$ & $\Delta u^\top \Delta y$       \\
Input Passivity $(\nu\ge0)$   & $0$                   & $-2\nu I$  & $I$ & $\Delta u^\top \Delta y -\nu | \Delta u |^2$         \\
Output Passivity $(\varepsilon\ge0)$ & $-2 \varepsilon I$    & $0$        & $I$ & $\Delta u^\top \Delta y -\varepsilon | \Delta y |^2$ \\ \bottomrule
\end{tabular}
\caption{Choices of $(Q,S,R)$ to verify different incremental properties.}
\label{tab:choice_Q_S_R}
\end{table}

To summarize, our goal is to learn the parameters $\theta$ of a dynamical system \eqref{eq:diff_system} — i.e., a Neural ODE — that optimize a given cost as per \eqref{eq:minimization_problem}, with the hard constraint that either \eqref{eq:definition_contracting_system} or \eqref{eq:definition_IQC_continuous_time} (or both) hold.
This new requirement (i.e., the hard constraint) must be satisfied for all the parameters $\theta$ we optimize over.
In other words, we consider a form of fail-safe learning, in the sense that the property of the model to be contracting or dissipative must be guaranteed during and after parameter optimization.
\begin{remark}
\label{re:contractive_integration}
In solving \eqref{eq:minimization_problem}, one must select a numerical integration scheme for simulating the trajectories of \eqref{eq:diff_system} to be used in forward propagation and gradient computations. However, the properties of the continuous-time model \eqref{eq:diff_system} may not be preserved in the discretized one, in general. While our focus in this paper is on the contractivity and dissipativity properties of the continuous-time model \eqref{eq:diff_system}, it is worth noting that integration methods preserving contractivity are discussed in \cite{manchester2017contracting}. Adapting these methods to IQC properties, or developing new ones, represents an interesting future research direction. %

\end{remark}
\section{Contractivity and Dissipativity of NodeRENs}
\label{sec:NODEREN}
In this section, we present and analyze novel DNN structures that arise from appropriately combining Neural ODEs~\cite{chenNeuralOrdinaryDifferential2019} with RENs~\cite{revay2023recurrent}. The starting idea is to define the functions $f_\theta$ and $g_\theta$ in \eqref{eq:diff_system} through the model utilized for the layer equation in \cite{revay2023recurrent}.
Specifically, we consider
\begin{gather}
\label{eq:NODEREN_linear_part}
    \begin{bmatrix}
    \Dot{x}(t) \\
    v(t) \\
    y(t)
    \end{bmatrix}
     = 
    \overbrace{{\begin{bmatrix}
    A & B_1 & B_2 \\
    C_1 & D_{11} & D_{12} \\
    C_2 & D_{21} & D_{22}
    \end{bmatrix}}}^{\tilde{A}}
    \begin{bmatrix}
    x(t) \\
    w(t) \\
    u(t)
    \end{bmatrix}
     + 
    \overbrace{\begin{bmatrix}
    b_x \\
    b_v \\
    b_y
    \end{bmatrix}}^{\tilde{b}},\\
    w(t) = \sigma ( v(t) )\label{eq:NODEREN_nonlinear_part},
\end{gather}
where \( x(t) \in \real{n}\), \(u(t) \in \real{m} \) and \(y(t) \in \real{p} \) are respectively the state, the input, and output at time $t$. 
The function $\sigma(\cdot)$ represents a nonlinear map and it is applied entry-wise.
The input and output of $\sigma(\cdot)$ are \(v(t), w(t) \in \mathbb{R}^q \), respectively. 
We denote system \eqref{eq:NODEREN_linear_part}-\eqref{eq:NODEREN_nonlinear_part} as a NodeREN, and it can be interpreted as an affine time-invariant system in closed-loop with a static nonlinearity $\sigma(\cdot)$.  
In NodeRENs, the set of trainable parameters $\theta \in \mathbb{R}^{n_\theta}$ consists of the set $\Tilde{A}$ of matrices  $(A$, $B_1$, $B_2$, $C_1$, $C_2$, $D_{11}$, $D_{12}$, $D_{21}$ and $D_{22})$ and the set $\tilde{b}$ of vectors $(b_x,b_v,b_y)$ in \eqref{eq:NODEREN_linear_part}. It is worth noting that the model \eqref{eq:NODEREN_linear_part}-\eqref{eq:NODEREN_nonlinear_part} is highly flexible, as it encompasses many important neural network architectures. For more details, see~\cite{revay2023recurrent}.

The work \cite{revay2023recurrent} has established conditions that ensure contractivity and IQC properties of the \emph{discrete-time}  dynamics induced by $\theta$. However, the same conditions on $\theta$ fail to ensure these properties for the continuous-time solutions of \eqref{eq:NODEREN_linear_part}-\eqref{eq:NODEREN_nonlinear_part}, in general. 

Towards establishing contractivity and IQC-based dissipativity in continuous-time for the model \eqref{eq:NODEREN_linear_part}-\eqref{eq:NODEREN_nonlinear_part}, consider two different possible trajectories of the system, starting from two initial conditions $a$, $b \in \real{n}$ and two input functions $u^{[1]}$, $u^{[2]} \in PC([0, \infty], \real{m})$.
Then, define the \textit{incremental form} of the system \eqref{eq:NODEREN_linear_part}-\eqref{eq:NODEREN_nonlinear_part}
\begin{gather}
\label{eq:NODEREN_INCREMENTAL_linear_part}
    \begin{bmatrix}
    \Delta\Dot{x}(t) \\
    \Delta v(t) \\
    \Delta y(t)
    \end{bmatrix}
    =  
    \begin{bmatrix}
    A & B_1 & B_2 \\
    C_1 & D_{11} & D_{12} \\
    C_2 & D_{21} & D_{22}
    \end{bmatrix}
    \begin{bmatrix}
    \Delta x(t) \\
    \Delta w(t) \\
    \Delta u(t)
    \end{bmatrix} \,,
    \\
    \Delta w(t) = \sigma ( v_b^{[2]}(t) + \Delta v(t) ) - \sigma(v_b^{[2]}(t)) \,,
    \label{eq:NODEREN_INCREMENTAL_nonlinear_part}
\end{gather}
where $\Delta x$, $\Delta y$ and $\Delta u$ are defined in \eqref{eq:delta_definition_y_u_x}. Moreover, $\Delta v(t) = v_a^{[1]}(t) -v_b^{[2]}(t)$, where $v_a^{[1]}$ and $v_b^{[2]}$ are the  inputs of $\sigma(\cdot)$ 
for each trajectory.
Next, we introduce two technical assumptions.
\begin{assumption}%
\label{assumption:rate_limited_sigma}
The function $\sigma(\cdot)$ belongs to  $PC([0,\infty],\real{})$ and its slope is restricted to the interval $[0,1]$, that is
\begin{equation*}
 \label{eq:rate_limited_sigma}
    0\leq \frac{ \sigma(y) - \sigma(x) }{y - x} \leq 1  \:, 
    \quad \forall x,y \in \mathbb{R} \:, 
    \quad  x \neq y.
\end{equation*}
\end{assumption}
It is important to notice that, under Assumption 1, $\Delta v(t)$ and $\Delta w(t)$  verify the following inequality
\begin{equation}
    \Gamma(t) = 
    \begin{bmatrix}
    \Delta v(t) \\
    \Delta w(t)
    \end{bmatrix}^\top
    \begin{bmatrix}
    0  &  \Lambda   \\
    \Lambda  &  -2\Lambda
    \end{bmatrix}
    \begin{bmatrix}
    *
    \end{bmatrix}
    \ge 0 ,\quad \forall t \in \mathbb{R},
    \label{eq:Gamma_conic_combination}
\end{equation}
for any diagonal matrix $\Lambda \succ 0$.
Note that most of the popular activation functions used in the literature, such as the logistic function, $ReLU(\cdot)$ and $\tanh{(\cdot)}$, satisfy \Cref{assumption:rate_limited_sigma}.

\begin{assumption}\label{assumption:D_11_lower_triangular}
    $D_{11}$ in \eqref{eq:NODEREN_linear_part} is strictly lower-triangular.
\end{assumption}

\Cref{assumption:D_11_lower_triangular} enforces that each scalar entry of $\Delta v(t)$ only depends on the ones above it  through \eqref{eq:NODEREN_INCREMENTAL_linear_part}. Hence, it becomes possible to explicitly compute $\Delta v(t)$. This simplifies the numerical calculation of the solutions to \eqref{eq:NODEREN_INCREMENTAL_linear_part}-\eqref{eq:NODEREN_INCREMENTAL_nonlinear_part}, while guaranteeing that the model is very expressive thanks to the recursive application of the nonlinearity $\sigma$ on successive entries of $\Delta v(t)$. The case where $D_{11}$ is not lower-triangular is left for future work. We refer the interested reader to~\cite{revay2023recurrent} for a discussion on implicit RENs in discrete-time.

In general, an ODE model may not admit a unique solution for a given initial condition and input trajectory \cite{lygeros2010lecture}. We now proceed to show that \eqref{eq:NODEREN_linear_part}-\eqref{eq:NODEREN_nonlinear_part} admits a unique solution for any choice of the parameters $\theta$.
\begin{lemma}%
\label{lemma:existence_uniqueness}
Let \Cref{assumption:rate_limited_sigma,assumption:D_11_lower_triangular} hold. Then, the model \eqref{eq:NODEREN_linear_part}-\eqref{eq:NODEREN_nonlinear_part} admits a unique solution in time for any $u(t) \in PC([0,\infty],\real{m})$, $x(0)\in \real{n}$ and for any choice of the parameters $\theta \in \mathbb{R}^{n_\theta}$.
\end{lemma}
The proof, reported in \preprintswitch{\cite{martinelli2023unconstrained}}{\Cref{appendix:proof_f_is_unique}} for completeness, is based on deriving a global Lipschitz constant by iterating through the nonlinearities defining the vector $v(t)$.

\subsection{Characterizing Contracting and Dissipative NodeRENs}
Here, we derive sufficient conditions over the parameters $\theta$ to guarantee contractivity and {dissipativity} of NodeRENs \eqref{eq:NODEREN_linear_part}-\eqref{eq:NODEREN_nonlinear_part} according to a specified supply rate \eqref{eq:definition_supply_rate_Q_S_R}. 
\begin{theorem}%
\label{theorem:NODEREN_LMI_contractivity}
A {NodeREN} \eqref{eq:NODEREN_linear_part}-\eqref{eq:NODEREN_nonlinear_part} is contracting according to \eqref{eq:definition_contracting_system}, if there exists a matrix $P\succ0$ and a diagonal matrix $\Lambda\succ0$ such that
\begin{gather}
\begin{bmatrix}
    -A^\top P - PA  &  -C_1^\top\Lambda - PB_1 \\
   * & W
\end{bmatrix}\succ 0 \,,\label{ineq:LMI_NODEREN_contractivity}
\\
W = 2\Lambda - \Lambda D_{11} - D_{11}^\top \Lambda. \label{eq:definition_W}
\end{gather}
\end{theorem}
\preprintswitch{}{The proof is reported in~\Cref{appendix:proof_theorem_contractivity}.} For the rest of the paper, we refer to NodeRENs that comply with \eqref{ineq:LMI_NODEREN_contractivity}-\eqref{eq:definition_W} as {C-NodeRENs}.
\begin{remark}
    One may want to impose a certain convergence rate for a {{C-NodeREN}}. For this purpose, \eqref{ineq:LMI_NODEREN_contractivity} can be modified as follows:
\begin{equation*}
\label{ineq:LMI_NODEREN_contractivity_with_alpha_star}
    \begin{bmatrix}
    -A^\top P -PA -2 \Tilde{\omega} P &  -C_1^\top\Lambda - PB_1 \\
    *  & W
    \end{bmatrix}
    \succ 0\,,
\end{equation*}
where  the term $-2\tilde{\omega}P$ enforces that \eqref{eq:definition_contracting_system} holds with $c=\tilde{\omega}$. This can be directly seen by following the proof of \Cref{theorem:NODEREN_LMI_contractivity}\preprintswitch{.}{ starting from \eqref{ineq:LMI_NODEREN_contractivity_with_alpha}.}
\end{remark}
Next, we characterize NodeRENs that comply with incremental IQCs.
\begin{theorem}%
\label{theorem:NODEREN_LMI_robustness}
A NodeREN \eqref{eq:NODEREN_linear_part}-\eqref{eq:NODEREN_nonlinear_part} satisfies the incremental IQC described by the triple $(Q,S,R)$ fulfilling \eqref{eq:Assumptions_Q_R}, if there exists a matrix $P \succ 0$ and a diagonal matrix $\Lambda \succ 0$ such that
\begin{multline}
\scalemath{0.93}{
    \begin{bmatrix}
    -A^\top P  -PA &  -C_1^\top\Lambda - PB_1  & -PB_2 + C_{2}^\top S^\top \\
    *  &  W  &    - \Lambda D_{12}+D_{21}^\top S^\top\\
    *  &  *  &  \mathcal{R} 
    \end{bmatrix}  } 
    \\
    \scalemath{0.93}{+
    \begin{bmatrix}
    C_2^\top \\ D_{21}^\top \\ D_{22}^\top
    \end{bmatrix}Q
    \begin{bmatrix}
    C_2^\top \\ D_{21}^\top \\ D_{22}^\top
    \end{bmatrix}^\top 
  \succ 0 \,,} \label{ineq:LMI_NODEREN_robustness}
\end{multline}
with $W$ given by \eqref{eq:definition_W} and
\begin{equation}\label{eq:def_matrix_U_Y_Z}
    \mathcal{R} = R + S D_{22} + D_{22}^\top S^\top \,.
\end{equation}
\end{theorem}

\preprintswitch{}{The proof is reported in~\Cref{appendix:proof_theorem_robustness}. }For the rest of the paper, we refer to NodeRENs that comply with \eqref{ineq:LMI_NODEREN_robustness}-\eqref{eq:def_matrix_U_Y_Z} as {{IQC-NodeRENs}}.
\begin{remark}
    Note that all {IQC-NodeRENs} are also contracting for any choice of $(Q,S,R)$ such that \eqref{eq:Assumptions_Q_R} holds. Indeed, if there exists $P\succ 0$ and a diagonal matrix $\Lambda \succ 0$ such that \eqref{ineq:LMI_NODEREN_robustness} holds, then \eqref{ineq:LMI_NODEREN_contractivity} also holds because $Q\preceq 0$ and the left-hand-side of \eqref{ineq:LMI_NODEREN_contractivity} is a principal minor of that of \eqref{ineq:LMI_NODEREN_robustness}. %
\end{remark}

It should be noted that solving the minimization problem \eqref{eq:minimization_problem} %
using the above parametrizations of
{C-NodeRENs} or {IQC-NodeRENs} requires solving a semidefinite program~(SDP) multiple times during the optimization process. This can be computationally intractable for DNNs with several parameters. In the following section, we propose a parametrization of a rich class of {C-NodeRENs} and {IQC-NodeRENs} that circumvents this problem and allows for unconstrained optimization. It is worth noting that this approach is similar to \cite{revay2023recurrent}, but the techniques differ due to the continuous-time nature of the proposed NodeRENs. %

\subsection{Direct Parametrization of NodeRENs}\label{subsec:direct_parametrization_NodeRENs}
Our goal is to define new free parameters $\theta_C \in \mathbb{R}^{n_C}$ and $\theta_{IQC} \in \mathbb{R}^{n_{IQC}}$ that can be mapped onto the parameters $\theta$ of the NodeREN model, given by equations \eqref{eq:NODEREN_linear_part}-\eqref{eq:NODEREN_nonlinear_part}, and such that either \eqref{ineq:LMI_NODEREN_contractivity} or \eqref{ineq:LMI_NODEREN_robustness} are verified. First, we focus on constructing C-NodeRENs from the parameters $\theta_C \in\mathbb{R}^{n_C}$ given by
\begin{equation}
    \label{eq:theta_C}
\theta_C= \{X, {B}_2, C_2, {D}_{12}, D_{21}, D_{22}, \Tilde{b}, U, Y_{1}, X_P\}\,,
\end{equation}
where $X \in\real{ (n+q) \times (n+q)}$, 
$B_2 \in\real{n\times m}$,
$C_2 \in\real{p\times n}$, 
$D_{12} \in\real{q \times m}$, 
$D_{21} \in\real{p \times q}$,
$D_{22} \in\real{p \times m}$,
$\tilde{b}\in\real{(n+q+p)}$
$U\in\real{n \times q}$,
$Y_1\in\real{n \times n}$ and 
$X_P\in\real{n \times n}$. 
\begin{theorem}%
\label{theorem:direct_parametrization_contracting}
    For any $\theta_C\in\real{n_{C}}$ defined in \eqref{eq:theta_C}, and $\epsilon,\epsilon_P>0$ the two following statements hold.
    \begin{enumerate}
        \item 
            {There are matrices $(Y,W,Z,P)$ of appropriate dimensions such that}
            \begin{gather}
                \begin{bmatrix}
                -Y^\top - Y  &  -U - Z \\
               * & W
            \end{bmatrix} = X^\top X + \epsilon I\,,\label{eq:theorem_contractivity_direct_parametrization}
            \\
            P = X_P^\top X_P + \epsilon_P I\,.  \label{eq:P_construction_pos_def}
            \end{gather}
        \item 
            {There are matrices
            $\{\tilde{A},\tilde{b}\}$ defined in terms of $\theta_C$ and the matrices $(Y,W,Z,P)$ defined in point $1$, such that the corresponding NodeREN  \eqref{eq:NODEREN_linear_part}-\eqref{eq:NODEREN_nonlinear_part} is contracting.}
    \end{enumerate}
    {Moreover, matrices $(Y,W,Z,P)$ and $\{\tilde{A},\tilde{b}\}$ can be computed as described in the proof of the Theorem.}
\end{theorem}

\preprintswitch{}{The proof is reported in~\ref{appendix:proof_theorem_direct_contractivity}. }In conclusion, an unconstrained parametrization of a class of {C-NodeRENs} is obtained by first freely choosing $\theta_C$ defined in \eqref{eq:theta_C},
and then recovering the matrix $\Tilde{A}$ \preprintswitch{following the steps reported in the proof (see \cite{martinelli2023unconstrained}).}{as per \eqref{eq:_y_y_top_eq_H_11}-\eqref{eq:A_B1_C1_from_U_Y_Z} in~\Cref{appendix:proof_theorem_direct_contractivity}.}

Next, we focus on NodeRENs that satisfy IQC constraints by design, i.e. for any choice of parameters  $\theta_{IQC} \in\real{n_{IQC}}$ given by
\begin{equation}
    \label{eq:theta_R}
    \theta_{IQC}= \{X_R,B_2,C_2,D_{21},\tilde{b},X_3,T,U,Y_{1}, X_{P}\}\,,
\end{equation}
where $X_R \in\real{ (n+q) \times (n+q)}$, 
$B_2 \in\real{n\times m}$,
$C_2 \in\real{p \times n}$, 
$D_{21} \in\real{p\times q}$, 
$\tilde{b}\in\real{(n+p+q)}$, 
$X_3 \in\real{s\times s}$, 
$T\in\real{q \times m}$, 
$U\in\real{n \times q}$,
$Y_1\in\real{n \times n}$ and 
$X_P\in\real{n \times n}$.
In the next theorem, we provide a procedure to construct an {IQC-NodeREN} for any choice of $\theta_{IQC}$.
{{\begin{theorem}\label{theorem:direct_parametrization_robust}
Let $(Q,S,R)$ be such that %
\eqref{eq:Assumptions_Q_R} holds.
Assume also that there exists $\delta>0$ satisfying $R-S(Q-\delta I)^{-1}S^{\top}\succ 0$. 
For any $\theta_{IQC}\in\real{n_{IQC}}$ defined in \eqref{eq:theta_R}, and $\epsilon,\epsilon_P>0$ the two following statements hold.
    \begin{enumerate}
        \item 
            There are matrices $(Y,W,Z,P,D_{22})$ of appropriate dimensions such that
            \begin{gather}
                \tilde{\mathcal{R}} = R + S D_{22} + D_{22}^\top S^\top + D_{22}^\top Q D_{22}\succ 0 \,,\label{eq:definition_R_cal}
                \\
                \begin{bmatrix}
                -Y^\top - Y  &  -U-Z \\
                *  &  W  
                \end{bmatrix}
                -
                \Psi
                =
                X_R^\top X_R + \epsilon I\,,\label{eq:theorem_robust_direct_parametrization}
            \end{gather}
            where $P$ is constructed as in \eqref{eq:P_construction_pos_def}, %
            and 
            \begin{equation}  
                \Psi
                =
                \begin{bmatrix}
                V \\ \Tilde{T}
                \end{bmatrix}
                \Tilde{\mathcal{R}}^{-1}
                \begin{bmatrix} * \end{bmatrix}^\top
                -
                \begin{bmatrix}
                C_{2}^\top \\ D_{21}^\top 
                \end{bmatrix}
                Q
                \begin{bmatrix} * \end{bmatrix}^\top\,,
                \label{eq:definition_Psi}
            \end{equation}
            \begin{gather}
                \Tilde{T} = - T  + D_{21}^\top S^\top + D_{21}^\top Q D_{22} \:, \label{eq:definition_T_tilde}\\
                V = -P B_{2} + C_{2}^\top S^\top + C_{2}^\top Q D_{22} \,.\label{eq:definition_V}
            \end{gather}
        \item
            There are matrices %
            $\{\tilde{A},\tilde{b}\}$
            defined in terms of $\theta_{IQC}$ and the matrices $(Y,Z,W,P,D_{22})$ defined in point~1, such that the corresponding {NodeREN} \eqref{eq:NODEREN_linear_part}-\eqref{eq:NODEREN_nonlinear_part} satisfies the incremental IQCs parametrized by $(Q,S,R)$. 
    \end{enumerate}
    {Moreover, matrices $(Y,W,Z,P)$ and $\{\tilde{A},\tilde{b}\}$ can be computed as described in the proof of the Theorem.}
\end{theorem}
}}
The proof of Theorem~\ref{theorem:direct_parametrization_robust} 
{shares similarities with the one of Proposition~2}
in \cite{revay2023recurrent} and it is reported in~\preprintswitch{\cite{martinelli2023unconstrained}}{\Cref{appendix:proof_of_direct_parametrization_robust}}.
Note that the assumption that there exists a $\delta$ such that $R-S(Q-\delta I)^{-1}S^{\top}\succ 0$ is not restrictive. Indeed, for all the most relevant cases, reported in \Cref{tab:choice_Q_S_R}, finding an appropriate value of $\delta$ is straightforward.
In conclusion, an unconstrained parametrization of {IQC-NodeRENs} is obtained by first freely choosing the parameters $\theta_{IQC}\in\real{n_{IQC}}$, then constructing \preprintswitch{$D_{22}$, $\Tilde{A}$ following the steps reported in the full proof.}{$D_{22}$ as per \eqref{eq:NODEREN_Last_D22_Robust} in Appendix~\ref{appendix:proof_of_direct_parametrization_robust}, and last recovering the matrices of $\tilde{A}$ using \eqref{eq:_y_y_top_eq_H_11}-\eqref{eq:A_B1_C1_from_U_Y_Z}, \eqref{eq:NODEREN_Last_D22_Robust} and \eqref{eq:T__Lambda_D_12}.} 
\section{Simulations \& Results}
\label{sec:Simulations_and_results}
In this section, we showcase the application of NodeRENs for identifying stable nonlinear systems. We begin by assessing the performance of NodeRENs when using various integration methods during training. Next, we highlight the advantages of contractivity-by-design by training a general non-contractive NodeREN \eqref{eq:NODEREN_INCREMENTAL_linear_part}-\eqref{eq:NODEREN_INCREMENTAL_nonlinear_part}, which fails to learn stable behavior for the same task. Finally, we demonstrate the ability of NodeRENs to learn from irregularly sampled time-series data while maintaining the guarantees by design. Our code is accessible at \url{https://github.com/DecodEPFL/NodeREN}.

We remark that NodeREN trajectories are time-reversible, as established in Lemma~\ref{lemma:existence_uniqueness}, which implies that they cannot intersect for distinct initial conditions. To enhance model expressiveness, we employ augmented state vectors, initializing additional scalar states to zero. This technique, known as feature augmentation, was introduced in \cite{dupont2019augmented} for general Neural ODEs.

\subsection{Continuous-Time System Identification}
Here, we consider the system identification of a black-box system.
For this experiment, we assume that the unknown system dynamics are those of a nonlinear pendulum
\begin{equation}
\label{eq:pendulum}
     \ell \Ddot{\alpha}(t) + \beta \Dot{\alpha}(t) + g \sin{\alpha(t)} =  0 \;,
\end{equation}
where $\alpha(t)$ is the angle position of the pendulum at time $t$ with respect to the vertical axis, $\beta>0$ is the viscous damping coefficient, $g$ is the gravitational acceleration and $\ell>0$ is the length of the pendulum. \preprintswitch{}{For our experiments, we have chosen $\ell = \SI{0.5}{\meter}$ and $\beta = \SI{1.5}{\meter\per\second}$.} 

{We consider the scenario where the system dynamics \eqref{eq:pendulum} to identify are completely unknown and the only prior knowledge we have is that the system is stable around the origin.}
{Hence, we choose to train a C-NodeREN, which is guaranteed to be contracting (and so, stable around the origin if $\Tilde{b}$ is null) 
for any choice of the trainable parameters $\theta_{C}$ --- i.e. 
even if we stop the optimization prematurely.}

The training data are given by noisy measurements of $\alpha(t)$ and $\dot{\alpha}(t)$ across $N \in \mathbb{N}$ different experiments performed in the time interval $[0,T_{end}]$. For each experiment\preprintswitch{}{~$i=1,\ldots,N$}, the system \eqref{eq:pendulum} starts from a different and known initial condition $\alpha_i(0)$ and $\dot{\alpha}_i(0)$. Due to the possibility of random time delays during data acquisition, we assume that trajectory measurements are taken at random time instants across $[0,T_{end}]$ for each experiment. 
\preprintswitch{We obtain the parameters $\theta_C$ minimizing the mean squared error between the C-NodeREN predictions and the simulated measurements. For more details, please refer to~\cite{martinelli2023unconstrained}.}{It is worth noting that inconsistent measurement times are common in control applications; however, incorporating this feature into discrete-time systems would pose a significant modeling challenge. 

Precisely,  the training data consist of $\mathcal{Z} = \{\mathcal{Z}_i\}_{i=1}^N$, where
\begin{equation}
\label{eq:dataset}
    \mathcal{Z}_i = \left\{(t_j,z_j), \text{ for all }t_j \in \mathcal{T}_i\subset[0,T_{end}]\subset\mathbb{R}\right\}\,,
\end{equation}
where $T_{end} = \SI{3}{\second}$, $\mathcal{T}_i$ contains $n_i\in \mathbb{N}$ time instants, and for each experiment $i=1,\ldots,N$, the vector $z_j \in \mathbb{R}^2$ is measured as $z_j= \begin{bmatrix}\alpha(t_j)&\dot{\alpha}(t_j)\end{bmatrix}^\top+w(t_j)$ with $w(t_j) \sim \mathcal{N}(0,0.01 I_2)$ being drawn according to a Gaussian distribution. We minimize over $\theta_C$ the loss function
\begin{equation*}
 L(\theta_C,\mathcal{Z})=\frac{1}{N} \sum_{i=1}^{N} 
 \frac{1}{n_i}
 \sum_{(t_j,z_j) \in \mathcal{Z}_i}
 |y(t_j)-z_j|^2 \,,
\end{equation*}
where, with slight abuse of notation, $y(t_j) \in \mathbb{R}^2$ denotes the output at time $t_j$  of the C-NodeREN with parameters $\theta_C$.

The chosen {C-NodeREN architecture {has $n= 4$ and $q=5$. It} comes with $\theta_C \in \mathbb{R}^{\sim 150 }$,
and it is trained using one of the following integration methods}: forward Euler (\texttt{euler}), Runge-Kutta of order 4 (\texttt{rk4}), and Dormand-Prince-Shampine of order 5 (\texttt{dopri5}).
The training is performed with $N =200$ experiments and using Adam~\cite{ruder2017overviewGradientDescent}.

{To test the prediction performance of our trained models, we construct a test dataset $\mathcal{Z}_{test}$ in the same way as \eqref{eq:dataset}, using a longer time window $[0,T_{end}^{test}]$ with $T_{end}^{test} = \SI{8}{\second}$ and new sets of time instants $\mathcal{T}_i^{test}$. We then compute the test loss $L(\theta_C^\star,\mathcal{Z}_{test})$, where $\theta_C^\star$ are the parameters where the training converged to.} %
}
In \Cref{fig:Plot_comparison_Integration_Methods}, we report the resulting test losses 
for different methods, along with their Number of Function Evaluations~(NFE) used for the prediction\preprintswitch{.}{ in the time window $[0,T_{end}^{test}].$
Both \texttt{euler} and \texttt{rk4} have been applied using $\{100, 200\}$ and $\{50, 200\}$ integration steps, respectively, in $[0,T_{end}^{test}]$. One of the advantages of variable-steps methods such as \texttt{dopri5} is the possibility to choose a tolerance value \emph{tol} to fix the computational error during the ODE simulation. Indeed, it is possible to obtain a trade-off between accuracy and the required number of function evaluations just by properly setting this parameter.}  
It is important to emphasize how for fixed-step methods (e.g., \texttt{euler}, \texttt{rk4}), the chosen number of steps may be insufficient to fully capture the dynamics of the trained model, which can result in numerical instability. In contrast, adaptive methods (e.g., \texttt{dopri5}) dynamically adjust the step size to ensure bounded integration errors, hence promoting numerical stability in the trajectory simulations. Such freedom enables us to either control the NFEs with a fixed step approach, or to achieve more accurate simulations with an adaptive step size method. As anticipated in Remark~\ref{re:contractive_integration}, it would be interesting to endow C-NodeRENs with specific integration methods that preserve contractivity, e.g.~\cite{manchester2017contracting}.

Then, for the same identification task, we have trained a general NodeREN  \eqref{eq:NODEREN_INCREMENTAL_linear_part}-\eqref{eq:NODEREN_INCREMENTAL_nonlinear_part} (that we denote as G-NodeREN), with trainable parameters $\theta = \{\tilde{A},\tilde{b}\}$. 
Given the higher flexibility of G-NodeRENs with respect to C-NodeRENs,  one could expect a better performance in the time window $[0,T_{end}]$ considered for the training. However, a trained G-NodeREN can exhibit unstable dynamics when simulating for longer horizons; this phenomenon is illustrated in  \Cref{fig:Plot_confidence_intervals}. 
Instead, C-NodeRENs only limit the search to stable models, resulting in better identification performance for the considered example. In \Cref{fig:Plot_confidence_intervals}, we also display trajectory tubes from different perturbed initial conditions. Note that the diameter of the tube associated with the C-NodeREN decreases over time, showcasing its contractivity.

\begin{figure}[t]
    \centering
    \includegraphics[width=0.33\textwidth]{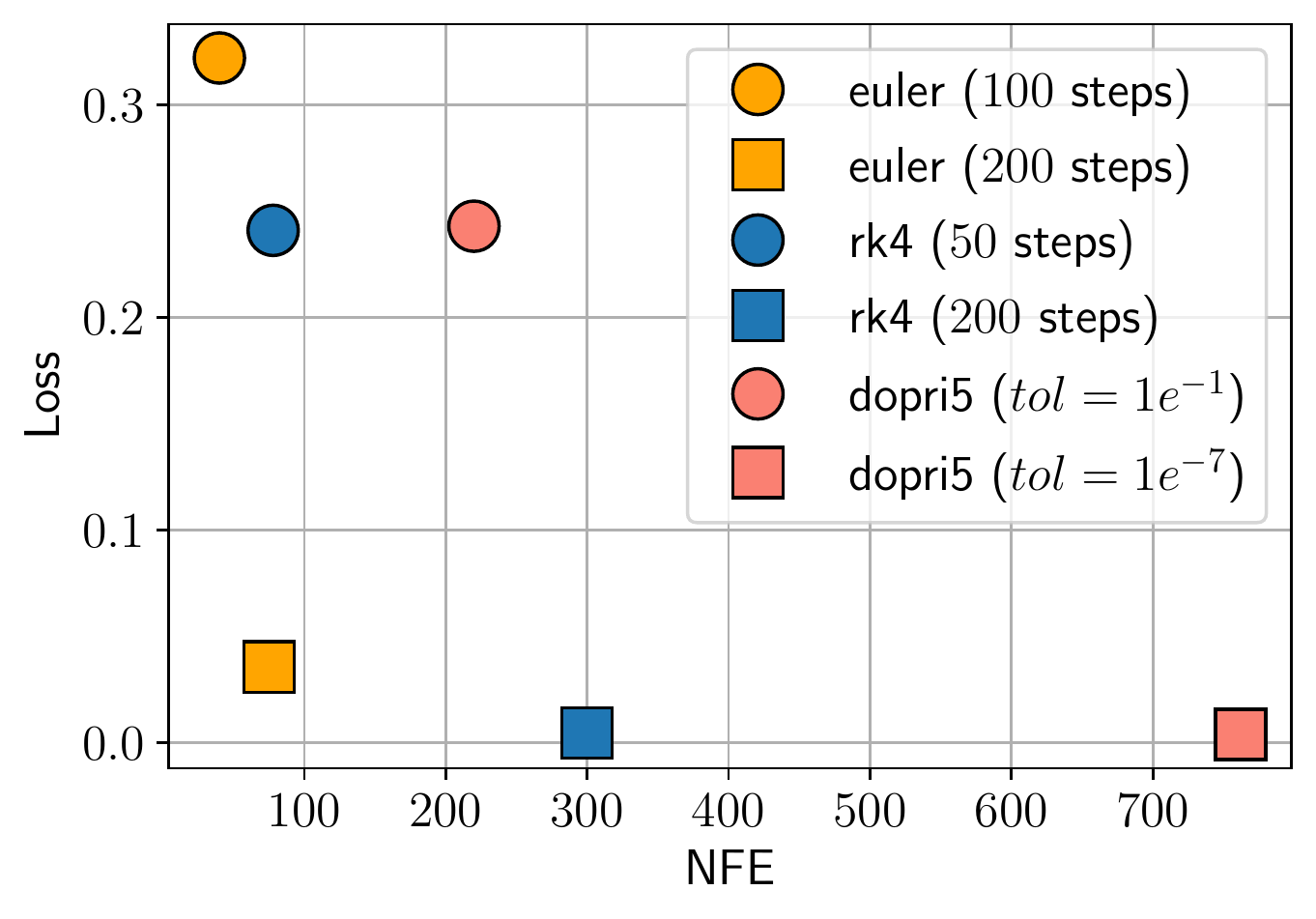}
    \caption{Comparison between the number of function evaluations~(NFE) and loss on testing dataset for three integration methods\preprintswitch{: forward Euler (\texttt{euler}), Runge-Kutta~IV (\texttt{rk4}), and Dormand-Prince-Shampine (\texttt{dopri5}).}{.}}
    \label{fig:Plot_comparison_Integration_Methods}
\end{figure}

\subsection{Irregularly sampled-data}\label{subsec:irregularly_sampled_data}
To show the robustness of C-NodeRENs concerning the irregular sample data, we compare the test losses {of $10$ C-NodeRENs} trained on $10$ different training datasets {having the same initial conditions, but different sampling times.} %
We observe that, despite using differently {sampled} data, all the models have similar test losses: they all lay in the interval $[3.7  \times 10^{-4} ; 9.1 \times 10^{-3}]$, depending on how well the samples were distributed.
\begin{figure}[t]
    \centering
    \includegraphics[width=0.35\textwidth]{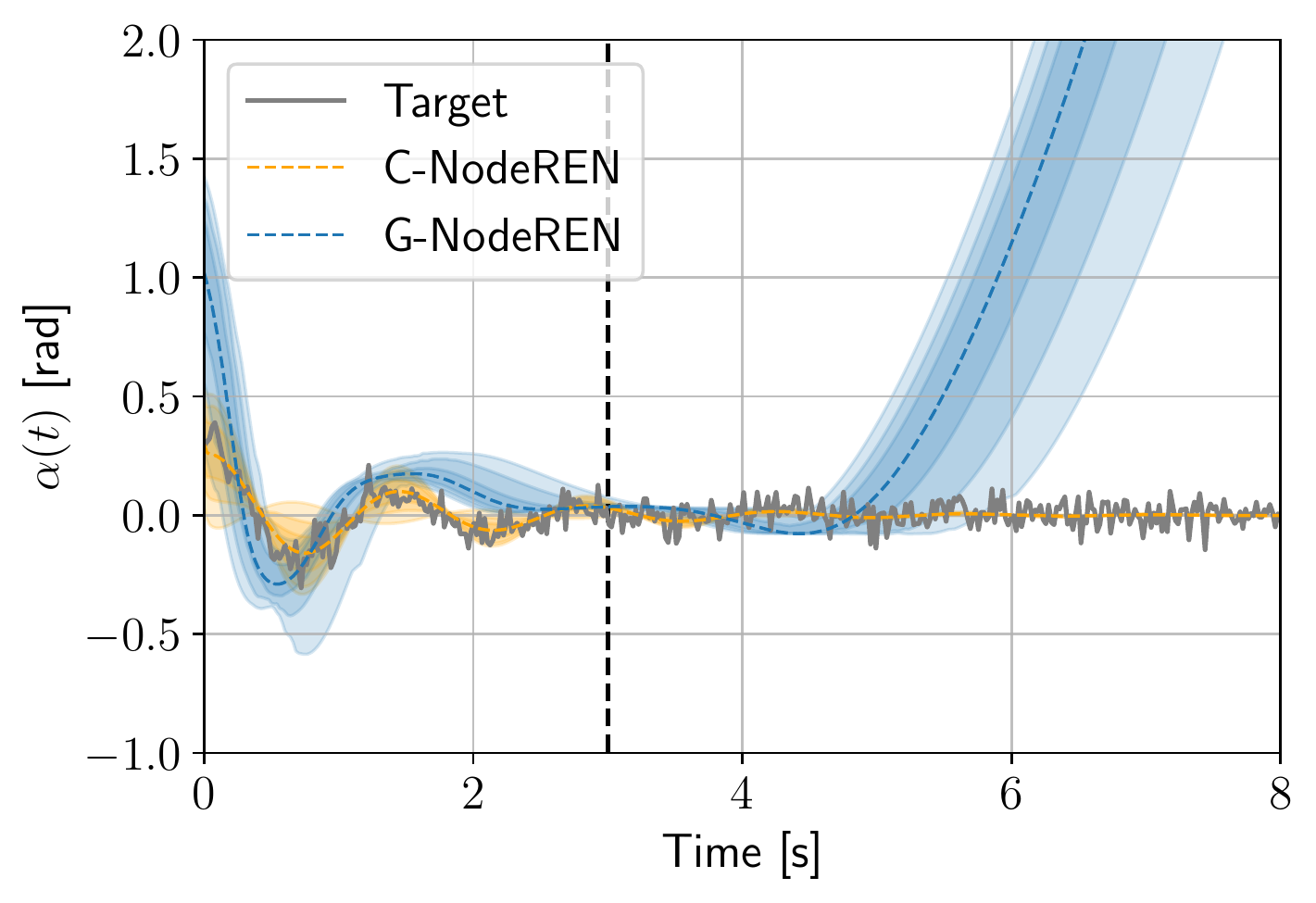}
    \caption{{Noisy test trajectory of the angular position $\alpha(t)$ of the real system (in gray). Predicted trajectories of the trained C-NodeREN (in orange) and G-NodeREN (in blue). Shaded areas represent tubes of additional trajectories resulting from perturbations on the initial state of the systems. The black dashed line indicates the end of the training horizon\preprintswitch{}{ $T_{end}=\SI{3}{\second}$}.}
    }
    \label{fig:Plot_confidence_intervals}
\end{figure}

\section{Conclusions}
\label{sec:Conclusions} 
In this work, we have established a class of Neural ODEs that generalize RENs for continuous-time scenarios. Specifically, the proposed models guarantee relevant continuous-time system-theoretic properties such as contractivity and dissipativity by design. The resulting DNNs can be trained using unconstrained optimization — which enables the use of large networks — and the resulting architectures are flexible in the choice of an ODE integration scheme that is suitable for the learning problem at hand. We have showcased the performance of NodeRENs through the task of identifying the behavior of a nonlinear pendulum, even with irregularly sampled data. 

Future research should focus on integration schemes preserving contractivity and IQC properties, exploring distributed NodeREN architectures, and analyzing their generalization capabilities.
\bibliography{bibliography}
\preprintswitch{}{
\appendix
\subsection{Proof of Lemma \ref{lemma:existence_uniqueness}}\label{appendix:proof_f_is_unique}
Given $u(t) \in PC([0,\infty],\real{m})$, we can rewrite \eqref{eq:NODEREN_linear_part} as: 
\begin{equation}\label{eq:definition_function_p}
    \dot{x}(t) = p(x(t),t) = A x(t) + B_1 w(t)+ B_2 u(t) + b_x, 
\end{equation}
where we define $p(x(t), t) = f(x(t), u(t), t)$. 
It is well-known that a solution exists and is unique if $p(\cdot,\cdot)$ is globally Lipschitz in its first argument and piece-wise continuous in its second (e.g., \cite[Theorem 3.6]{lygeros2010lecture}). Therefore, we prove that these properties hold for the model in \eqref{eq:NODEREN_linear_part} under \Cref{assumption:rate_limited_sigma,assumption:D_11_lower_triangular}.
{First, $p(\cdot,\cdot)$ is piece-wise continuous in its second argument because it is the composition of piece-wise continuous functions under \Cref{assumption:rate_limited_sigma}.}
Then, we prove that $p(\cdot,\cdot)$ is globally Lipschitz in its first variable, that is
\begin{multline*}
    \exists \kappa>0 : 
    |p(x_1,t) - p(x_2,t) | \le \kappa |x_1-x_2| \,, 
\end{multline*}
for every $x_1,x_2 \in \real{n}$ and $t \in \mathbb{R}$.
    For $j\in\{1,2\}$, denote $w_j \in \mathbb{R}^q$ as 
    {
    $w_j = \left[w_j^1,\ldots,w_j^q\right]^{\top}=\sigma(v_j)=\left[\sigma\left(v_j^1\right),\ldots,\sigma\left(v_j^q\right)\right]^{\top}$%
    },
    where $v_j=C_1x_j+D_{11}w_j+D_{12}\hat{u}$ for any possible value $\hat{u}$ of the function $u(t)$ in time.
    {By using the triangular inequality, we obtain
    \begin{align}
        |p(x_1,t) - p(x_2,&t) | =
        | A (x_1 - x_2) + B_1 (w_1 - w_2)|\nonumber\\
        &\leq | A | \,| x_1 - x_2| + | B_1| \, |w_1 - w_2|\,.
        \label{eq:only_step_global_lipschitz}
    \end{align}}
    Using \Cref{assumption:rate_limited_sigma} and \Cref{assumption:D_11_lower_triangular} we verify that $|w_1^1 - w_2^1|  \le 
        |C_1^{1:}| |x_1 - x_2 |$, where $C_1^{i:}$ denotes the $i^{th}$ row of $C_1$, and we define $\kappa_1 = |C_1^{1:}|$. 
    {Proceeding similarly for each entry of $w_1-w_2$ we derive that, for every $r=2,3,\ldots,q$, it holds that $|w_1^r-w_2^r|\leq \kappa_r|x_1-x_2|$, with Lipschitz constants $\kappa_r = | C_1^{r:} | + \sum_{\ell=1}^{r-1}|D_{11}^{r,\ell}| \kappa_\ell$, where $D_{11}^{{i},j}$ denotes the element in the ${i}^{th}$ row and $j^{th}$ column of $D_{11}$.}
    Finally, from \eqref{eq:only_step_global_lipschitz} we conclude $|p(x_1,t) - p(x_2,t) | \le \kappa_{tot} |x_1 - x_2 |$, where $\kappa_{tot} = |A| + \sum_{i = 1}^{q} |B_1^{i:}| \kappa_i$.
    Thus, the differential equation $\dot{x}(t) = f(x{(t)},u(t),t) = p(x{(t)},t)$ is globally Lipschitz in $x$ with constant $\kappa_{tot} \in \real{+}_0$ and its solution exists and it is unique.%

\subsection{Proof of Theorem~\ref{theorem:NODEREN_LMI_contractivity}}\label{appendix:proof_theorem_contractivity}
Define the function
\begin{equation}\label{eq:definition_V_delta}
    V(\Delta x(t)) = \Delta x(t)^\top P \Delta x(t) \,.
    \end{equation}
Let
\begin{equation*}
    \Phi = (-A^\top P  - PA) - (C_1^\top\Lambda + PB_1) W^{-1} (*)^\top \succ 0,
\end{equation*}
be the Schur complement of \eqref{ineq:LMI_NODEREN_contractivity}. 
Then, using properties of matrix inequalities, it is always possible to find an $\omega \in (0,{\lambda_{min}(\Phi)}\:{\lambda_{max}(P)}^{-1}]$
such that 
\(
    (\Phi - \omega P) \succ 0.
\)
Hence,
\begin{equation}
    \begin{bmatrix}
    \label{ineq:LMI_NODEREN_contractivity_with_alpha}
    -A^\top P  - PA -\omega P &  -C_1^\top\Lambda - PB_1 \\
        *  & W
    \end{bmatrix}
    \succ 0.
\end{equation}
By left- and right-multiplying the inequality \eqref{ineq:LMI_NODEREN_contractivity_with_alpha} with
\(
\begin{bsmallmatrix}
\Delta x(t)^\top & \Delta w(t)^\top
\end{bsmallmatrix}
\)
 and
\(
\begin{bsmallmatrix}
\Delta x(t)^\top & \Delta w(t)^\top
\end{bsmallmatrix}^\top
\), respectively, and using \eqref{eq:Gamma_conic_combination}, we obtain by substitution that
\begin{equation*}
    \Dot{V}(\Delta x(t)) \: +\omega V(\Delta x(t))< - \Gamma(t) \le 0 \;.
\end{equation*}
Thus, according to Lyapunov's exponential stability theorem~\cite{Khalil:1173048}, the incremental system is globally exponentially stable, that is, it satisfies \eqref{eq:definition_contracting_system} with $c = \nicefrac{\omega}{2}$.

\subsection{Proof of Theorem~\ref{theorem:NODEREN_LMI_robustness}}\label{appendix:proof_theorem_robustness}
From \Cref{def:Incremental_IQC}, if the function $\mathcal{S}(\Delta x(t))$ is continuous and differentiable, \eqref{eq:definition_IQC_continuous_time} can be rewritten as:
\begin{equation}\label{eq:differentiated_dissipation_inequality}
     \frac{d}{dt}\Bigl(\mathcal{S}(\Delta x(t)) \Bigl) \: \le \: s_\Delta\bigl( \Delta u(t), \Delta y(t)\bigl) \;.
\end{equation}
Consider as $\mathcal{S}(\Delta x(t))$ the function $V(\Delta x(t))$ as in \eqref{eq:definition_V_delta}, which is continuous and differentiable in $\Delta x(t)$. By left- and right-multiplying the inequality \eqref{ineq:LMI_NODEREN_robustness} with
\(
\begin{bsmallmatrix}
\Delta x(t)^\top & \Delta w(t)^\top & \Delta u(t)^\top
\end{bsmallmatrix}
\)
 and
\(
\begin{bsmallmatrix}
\Delta x(t)^\top & \Delta w(t)^\top & \Delta u(t)^\top
\end{bsmallmatrix}^\top
\) 
respectively, we obtain by substitution 
\begin{equation*}
    \Dot{V}(\Delta x(t)) -
    \begin{bmatrix}
    \Delta y(t) \\ \Delta u(t) 
    \end{bmatrix}^\top
    \begin{bmatrix}
    Q  &  S^\top \\
    S  &  R
    \end{bmatrix}
    \begin{bmatrix}
    *
    \end{bmatrix}
    <
    - \Gamma(t)
    \le 0\,.
\end{equation*}
Thus, the inequality \eqref{eq:differentiated_dissipation_inequality} holds.

\subsection{Proof of Theorem~\ref{theorem:direct_parametrization_contracting}}\label{appendix:proof_theorem_direct_contractivity}
\emph{Part 1)}: 
    Define
\begin{equation}\label{eq:definition_H_contractivity}
    H = 
    \begin{bmatrix}
    H_{11} & H_{12} \\
    H_{21} & H_{22}
\end{bmatrix}
=
X^\top X + \epsilon I \,,
\end{equation}
where $H_{11} \in \real{n \times n}$, $H_{12} \in \real{n \times q}$, $H_{21} \in \real{q \times n}$ and $H_{22} \in \real{q \times q}$. 
Then, one can set
\begin{equation}\label{eq:_y_y_top_eq_H_11}
    Y = -\frac{1}{2}(H_{11} + Y_{1} - Y_{1}^\top )\,,
\end{equation}
\begin{equation}\label{eq:W_from_H_22_Z_from_H_12}
    W =  H_{22}\,, \quad Z = - H_{12} - U\,.
\end{equation}
Thus, starting from $H$ as in \eqref{eq:definition_H_contractivity}, $(Y,W,Z)$ can be built using \eqref{eq:_y_y_top_eq_H_11}-\eqref{eq:W_from_H_22_Z_from_H_12} and $P$ using \eqref{eq:P_construction_pos_def}.

\emph{Part 2)}: %
First, note that $H\succ0$ and $P\succ 0$ are guaranteed by construction because $X^{\top}X+\epsilon I$ and $X_P^{\top}X_P+\epsilon_P I$ are positive definite for any $X$ and $X_P$.
From \Cref{assumption:D_11_lower_triangular}, ${D}_{11}$ is strictly lower triangular. 
Thus, using \eqref{eq:definition_W} and being $W = W^\top$ by construction, it is possible to split $W$ as: $W= W_{diag} + W_{low} + W_{low}^\top$, where $W_{diag}$ is a diagonal matrix, and $W_{low}$ is a strictly lower triangular matrix.
By setting
\begin{equation}\label{eq:Lambda_and_D11_from_W}
    \Lambda = \frac{1}{2} W_{diag} \,,\quad D_{11} = - \Lambda^{-1}W_{low}\,,
\end{equation}
and the matrices $A, B_1, C_1$ can be retrieved 
as
\begin{equation}\label{eq:A_B1_C1_from_U_Y_Z}
    C_1 = \Lambda^{-1} U^\top  \,,\quad A = P^{-1}Y \,,\quad B_1 = P^{-1}Z \,,
\end{equation}
where $P^{-1}$ and $\Lambda^{-1}$ exist since $P\succ0$ and $\Lambda\succ 0$.
Finally, we can show that \eqref{ineq:LMI_NODEREN_contractivity} holds by substituting $A$, $B_1$, $C_1$, $D_{11}$, $\Lambda$, and $P$ with their definitions in 
\eqref{eq:definition_H_contractivity}-\eqref{eq:A_B1_C1_from_U_Y_Z}. Thus, the corresponding {NodeREN} \eqref{eq:NODEREN_linear_part}-\eqref{eq:NODEREN_nonlinear_part} specified by $\{\Tilde{A},\Tilde{b}\}$
is contracting. 

\subsection{Proof of Theorem~\ref{theorem:direct_parametrization_robust}}
\label{appendix:proof_of_direct_parametrization_robust}

We first need a technical lemma.
\begin{lemma}\label{lemma:cayley_transform}
    Let $p, m \in \mathbb{N}$, and $s = \max (p,m)$.
    Let $M \in \real{s \times s}$, and $M = M^\top \succ 0$.
    Define
    \begin{equation}\label{eq:F_cayley}
        F = (I - M)(I+M)^{-1}\,,
    \end{equation}
    and 
    \(\Tilde{F} = [F]_{p \times m}\,.\)
    Then it follows that $I - \Tilde{F}^\top \Tilde{F} \succ 0$.
\end{lemma}
  \begin{proof}  
Notice that $F$ as defined in \eqref{eq:F_cayley} is the Cayley transformation of $M$.
First, we list several properties of the Cayley transformation.
\begin{enumerate}
    \renewcommand{\labelenumi}{{\theenumi}}
    \renewcommand{\theenumi}{P\arabic{enumi}}
    \item \label{prop1} $M = M^\top \Rightarrow F=F^{\top}$,
    \item \label{prop2} $I+F = 2(I+M)^{-1}$,
    \item \label{prop3} $I-F = 2(I+M)^{-1}M$.
    \setcounter{properties_cayley}{\value{enumi}}
\end{enumerate}
Moreover, note that the following properties hold for any symmetric matrices $M,N$:
\begin{enumerate}
    \setcounter{enumi}{\value{properties_cayley}}
    \renewcommand{\labelenumi}{{\theenumi}}
    \renewcommand{\theenumi}{P\arabic{enumi}}
    \item \label{prop4} If $M\succ0$, $N\succ0$ and $MN = NM$, then $MN\succ0$,
    \item \label{prop5} $(I+M)^{-1}M = M(I+M)^{-1}$.
\end{enumerate}
We are now ready to prove \Cref{lemma:cayley_transform} based on these properties.

Let $M$ be such that $M=M^\top \succ 0$. First, we prove that  
\begin{equation}
\label{eq:cayley_part1}
    I-F^{\top}F\succ 0\,.
\end{equation}
 This is equivalent to proving 
\begin{equation}\label{eq:cayley_I+F_I-F}
(I+F)^{\top} (I-F) \succ 0\,,
\end{equation}
since $F = F^{\top}$. Using \ref{prop4}, we prove \eqref{eq:cayley_I+F_I-F} in three steps.
\begin{enumerate}
    \item From property~\ref{prop2}, $(I+F)\succ0$ since $2(I+M)^{-1}\succ 0$ {thanks to $M\succ 0$}.
    \item From property~\ref{prop3}, $(I-F)\succ0$ since $2(I+M)^{-1}M\succ 0$ from properties~\ref{prop4} and \ref{prop5}.
    \item We prove that $(I+F)^\top$ and $(I-F)$ are commutative:
    \begin{align*}
        (I+F)^{\top} (I-F) &= I-F + F^{\top} + F^{\top}F \\
        &= I+ F^{\top} - F + FF^{\top} \\
        &= (I-F) (I+F)^{\top}
    \end{align*}
    where the second equality is valid from \ref{prop1}.
\end{enumerate}
Then, \eqref{eq:cayley_I+F_I-F} holds and this proves that $I-F^{\top}F\succ0$. 

Last, we prove that \eqref{eq:cayley_part1} implies $I-\tilde{F}^\top\tilde{F}\succ 0$. Note that if $p=m$ then $\tilde{F}=F$ and this fact is trivial. We consider the two remaining cases $p>m$ and $m>p$ separately. 

Assume that $p>m$. Then, $F= \begin{bmatrix}\tilde{F} \\ F_b \end{bmatrix}$.
Hence, $I-F^{\top}F = I-\tilde{F}^{\top}\tilde{F} - F_b^{\top}F_b \succ 0$, which implies that $I-\tilde{F}^{\top}\tilde{F} \succ F_b^{\top}F_b \succeq 0$.

Assume that $m>p$. Then, $F= \begin{bmatrix}\tilde{F} & F_r \end{bmatrix}$. Hence,
\begin{equation*}
    I-F^{\top}F 
    =
    \begin{bmatrix}
        I-\tilde{F}^{\top}\tilde{F} & -\tilde{F}^{\top}F_r \\
        -F_r^{\top}\tilde{F} & I-F_r^{\top}F_r
    \end{bmatrix}
    \succ 
    0\,,
\end{equation*}
and by applying the Schur complement, the above equation implies $I-\tilde{F}^{\top}\tilde{F}\succ 0$.
\end{proof}

{Note that the property \eqref{eq:F_cayley} has also been used in~\cite{revay2023recurrent}. In this work, we have derived a full proof for completeness.}  We are now ready to present the proof of Theorem~\ref{theorem:direct_parametrization_robust}.

\begin{proof}
\emph{Part 1)}: {Define $\mathcal{Q} = Q-\delta I$. Since $\mathcal{Q}\prec0$ and $R-S(Q-\delta I)^{-1}S^{\top}\succ 0$, there exist $L_Q\succ 0$ and $L_R \succ 0$ with}
\begin{equation}\label{eq:L_Q_and_L_R}
    L_{Q}^\top L_{Q} = -\mathcal{Q} \:, \quad  L_{R}^\top L_{R} = R - S \mathcal{Q}^{-1} S^\top\,.
\end{equation}
Next, define $M = X_{3}^\top X_{3} + \epsilon I$ and
$\Tilde{F} = \left[(I-M)(I+M)^{-1}\right]_{p\times m}$. %
Construct $D_{22}$ as
\begin{equation}\label{eq:NODEREN_Last_D22_Robust}
    D_{22} = -\mathcal{Q}^{-1}S^\top + L_{Q}^{-1} \tilde{F} L_{R}.
\end{equation}
{We now show that this choice of $D_{22}$ leads to $\tilde{\mathcal{R}}\succ 0$. 
Let $\Xi = R + S D_{22} + D_{22}^\top S^\top + D_{22}^\top \mathcal{Q} D_{22}$. Since $\tilde{\mathcal{R}}=\Xi+\delta D_{22}^\top D_{22}$, it suffices to prove that $\Xi \succ 0$.}
After plugging \eqref{eq:NODEREN_Last_D22_Robust} in the definition of $\Xi$, we obtain $\Xi = L_R^{\top}(I-\Tilde{F}^{\top}\Tilde{F})L_R$
 by cancellation of addends and by using  \eqref{eq:L_Q_and_L_R}.
By \Cref{lemma:cayley_transform} and since $L_R\succ0$, we conclude that $\Xi \succ 0$.

Finally, we indicate how to construct $(Y,W,Z)$ so that \eqref{eq:theorem_robust_direct_parametrization} holds. Compute $\Psi$ from \eqref{eq:definition_Psi} and define
\begin{equation*}
H 
=
\begin{bmatrix}
    H_{11} & H_{12} \\
    H_{21} & H_{22}
\end{bmatrix}
=
X_R^\top X_R + \epsilon I + \Psi \,.
\end{equation*}
{Then, using $H$, it is possible to use the same steps of \Cref{theorem:direct_parametrization_contracting} to} compute $(Y,W,Z)$ from \eqref{eq:_y_y_top_eq_H_11}-\eqref{eq:W_from_H_22_Z_from_H_12}, and $P$ from \eqref{eq:P_construction_pos_def}.
This choice satisfies the equality \eqref{eq:theorem_robust_direct_parametrization} by construction.

\emph{Part 2)}: 
{We show how to parametrize the matrices in $\tilde{A}$.
First, note that $P\succ 0$ is valid due to \eqref{eq:P_construction_pos_def}. Construct $\Lambda$ and $D_{11}$ using \eqref{eq:Lambda_and_D11_from_W}, 
and define the matrices $(A,B_1,C_1,D_{11})$ according to \eqref{eq:A_B1_C1_from_U_Y_Z}. Set 
\begin{equation}
    \label{eq:T__Lambda_D_12}
    D_{12} = \Lambda^{-1} T\,.  
\end{equation}
Then, when choosing the parameters according to part~1 of this Theorem and plugging in the matrices just defined, $\Psi$ %
satisfies 
$\Psi= \Psi^{\top}\succeq 0$ since $\tilde{\mathcal{R}}\succ0$ and $Q\preceq0$. Moreover, $H - \Psi = X_R^{\top} X_R + \epsilon I \succ 0$. 
Thus, $H=H^{\top} \succ 0$, and we have $W = W^{\top} \succ 0$ and $\Lambda \succ 0$. } 
Thus, the {NodeREN} defined by $\{\tilde{A},\Tilde{b}\}$ complies with \eqref{ineq:LMI_NODEREN_robustness}, and thus it satisfies the incremental IQCs parametrized by $(Q,S,R)$. 
\end{proof}}
\end{document}